\documentclass{article}

\usepackage{graphicx}
\usepackage{amssymb}
\usepackage{amsthm}

\newtheorem{thm}{Theorem}
\newenvironment{definition}[1][Definition.]{\begin{trivlist}
\item[\hskip \labelsep {\bfseries #1}]}{\end{trivlist}}

\title{On a variant of Monotone NAE-3SAT and the Triangle-Free Cut problem.}
\author{Peiyush Jain, Microsoft Corporation.}

\begin{document}
\maketitle

    \abstract{
In this paper we define a restricted version of Monotone NAE-3SAT and show that it remains NP-Complete even under that restriction. We expect this result would be useful in proving NP-Completeness results for problems on $k$-colourable graphs ($k \ge 5$). We also prove the NP-Completeness of the Triangle-Free Cut problem.
    } 

\section{Introduction}

SAT was first proven to be NP-Complete in a seminal paper by Cook \cite{Cook} which led to the current theory around NP-Completeness.  Many variants of SAT have been shown to be NP-Complete and the Not-All-Equal 3SAT \cite{Schaeffer} (or NAE-3SAT in short) has been found to be particularly useful.
\\

SAT variants have been especially useful in proving NP-Completeness of various graph related problems for restricted classes of graphs. For instance, the planar version of 3SAT, Planar-3SAT is known to be NP-Complete \cite{Planar3SAT} and has been useful in proving NP-Completeness of planar graph related problems. Unfortunately, the planar version of NAE-3SAT is in P \cite{Moret}.
\\

 We consider a version of NAE-3SAT which might prove useful in proving NP-Completeness results of $k$-colourable graphs, for $k \ge 5$. For instance, from the result of this paper, it immediately follows that MAXCUT is NP-Complete for $5$-colourable graphs. Of course, this is weaker than the currently known result that MAXCUT is NP-Complete for $3$-colourable graphs \cite{MaxCut}, but we expect this result might still prove to be useful for other problems.
\\

We also show that the triangle-free cut problem is NP-Complete. This is the vertex colouring variant of the Monochromatic Triangle Problem, which is known to be NP-Complete \cite{GJ}.
\section {Preliminaries}

In this section we define the problems about which we are going to prove NP-Completeness results.

\subsection{$k$-colourable Monotone NAE-3SAT}
Given any monotone CNF expression $\phi$ with exactly three distinct (unnegated) literals per clause, we associate a simple undirected graph $G(\phi)$ with $\phi$ as follows:\\

Each variable of $\phi$ is a vertex of the graph, and two vertices have an edge between them iff they appear in some clause together. A variant is to have an additional vertex for each clause and connect them to the variables which appear in them. The result of this paper holds true for both the variants.
\\

We can now define our problem.
\begin{definition}{The {\bf $k$-colourable Monotone NAE-3SAT decision problem}: {\it Given a monotone CNF expression $\phi$ with exactly three distinct variables in each clause, such that the corresponding graph $G(\phi)$ defined earlier is $k$-colourable, is the expression $\phi$ not-all-equal satisfiable?}}
\end{definition}

It is easy to see that $4$-colourable monotone NAE-3SAT is in P.
\\

We show that $5$-colourable Monotone NAE-3SAT is NP-Complete, by going via a reduction that also proves the NP-Completeness of the Triangle-Free cut problem (defined below).
\\

\subsection {Triangle Free Cut}

Suppose we are given a graph and we want to know if there is a colouring of the vertices with two colours so that there is no monochromatic triangle. This problem can be stated in terms of cuts and triangle-free graphs.
\\

\begin{definition}{ {\bf Triangle-Free Cut}: Given a simple, undirected graph $G(V,E)$ with vertex set $V$ and edge set $E$, a cut of G is a partition of $V$ into two non-empty sets $V_1$ and $V-V_1$. We say that the cut is {\it triangle-free} if the induced sub-graphs of G on $V_1$ and $V-V_1$ both do not contains any triangles.
}
\end{definition}

 {
Thus we have:
\begin{definition}{The {\bf Triangle-Free Cut decision problem}: {\it Given a simple, undirected graph $G(V,E)$, is there a triangle-free cut of $G$?} }
\end{definition}

 {
It is easy to see that this problem is in NP.
}
\\

 {
 For the class of $4$-colourable graphs, the problem of determining the existence of a triangle-free cut is trivial: Just form a cut by putting vertices of two colours in one set and the remaining coloured vertices in the other.
}
\\

We show that Triangle-Free Cut problem is NP-Complete for the class of $5$-colourable and bounded degree ($\leq 8$) graphs.

\section{NP-Completeness of Triangle-Free Cut}

We start with the version of Monotone-NAE-3SAT in which each clause has exactly three distinct variables \cite{MoretBook}. Given such a Monotone-NAE-3SAT expression, we construct a graph which has a triangle free cut if and only if the expression is NAE-satisfiable. We show that the constructed graph has the property that it is $5$-colourable and has maximum degree $\leq 8$.
\\
i.e. we prove the following theorem:

\begin{thm}
The Triangle-Free Cut decision problem is NP-Complete for the class of $5$-colourable and maximum degree $\le 8$ graphs.
\end{thm}

\begin{proof}
Suppose we are given a monotone expression with variables $x_{1}, x_{2}, ..., x_{n}$ and clauses $c_{1},..., c_{m}$. We first transform this expression as follows:
\\

If a variable $x = x_1$(say) appears $k > 1$ times, we replace it with $k$ variables $y_{j}, 1 \leq j \leq k$ and add $k-1$ clauses to ensure that the $y_j$ variables get the same truth values.
\\

wlog, assume that $x = x_1$ appears in clauses $c_1, c_2, ..., c_k$.

Say $c_i = (x \vee a_i \vee b_i)$, $1 \leq i \leq k$ 

We replace these clauses with the $k$ new clauses: 

$c_{i}^{'} = (y_i \vee a_i \vee b_i)$, $1 \leq i \leq k$ 

and add $k-1$ new clauses 

$c_{i}^{''} = (y_i \vee \neg y_{i+1}), 1 \leq i \leq k-1$
\\

We do this for each repeated variable. This transformation is a polynomial time transform and the resulting expression is NAE-satisfiable if and only if the original expression is NAE-satisfiable.
\\

The resulting expression has the following properties:\\
\begin{enumerate}
\item Each clause is of the form $(x \vee y \vee z)$ or $(x \vee \neg y)$.
\item Each variable appears no more than three times.
\item Any two variables appear together in no more than one clause.
\item A variable appears exactly once in a clause with three literals. 
\item If a variable appears exactly once or twice, it appears un-negated.
\item If a variable appears thrice, the negated literal appears exactly once.
\end{enumerate}

We now construct a graph $G(V,E)$ from the expression we just created.
\\

We associate a vertex with each variable. For any clause with exactly three variables (note that the literals here are unnegated), we connect the corresponding vertices, forming a triangle. Observe that no two such triangles have a common vertex or edge.
\\

For clauses of the form $(x \vee \neg y)$, we 'connect' them using the glued-tetrahedra gadget in figure 1.
\\
\begin{figure}[h]
    \caption{Tetrahedra Gadget}
    \begin{center}
        \includegraphics[scale=0.9]{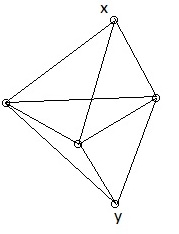}
    \end{center}
\end{figure}

In the gadget, notice that $x$ and $y$ are not adjacent, of degree 3 and that it is $5$-colourable even if we assign arbitrary colours to $x$ and $y$. Also notice that any triangle-free cut of this gadget necessarily has $x$ and $y$ in the same set and that there actually exists a triangle-free cut of this gadget.
\\

Now the graph so formed is $5$-colourable:\\
We pick a three literal clause and colour the corresponding vertices of the triangle arbitrarily with three different colours. Given any clause of the form $(x \vee \neg y)$, we already have a colouring of $x$ and $y$ and the vertices in the gadget connecting them can be coloured so as to maintain 5-colourability.
\\

Also notice that the maximum degree of any vertex in this graph is $8$, as it can be part of atmost two gadgets and one triangle.
\\

Now consider any triangle in the graph. This is either a triangle formed due to a three variable clause, or is a triangle completely within a gadget instance.
\\

We now claim that this graph has a triangle-free cut if and only if the expression is NAE-Satisfiable.
\\
\subsection{NAE-Satisfiablity implies existence of triangle-free cut}
Assume the expression was NAE-Satisfiable. Now put the variables which are set to true in one set and the variables which are set to false in another.  We now show that this can be extended to give a triangle-free cut of the graph.
\\

For any triangle which is formed because of a clause, since not all variables are true (or false), at least one edge of the triangle is part of the cut.
\\

Any clause of the form $(x \vee \neg y )$ is NAE-Satisfiable iff both $x$ and $y$ get assigned true, or both get assigned false. Now this implies that the corresponding gadget vertices can be assigned to the vertex sets so that all the triangles in the gadget are 'cut' by the cut.
\\

Thus, if the expression is NAE-Satisfiable, then the graph has a triangle-free cut.
\subsection{Existence of triangle-free cut implies NAE-Satisfiability}

Now assume that the graph has a triangle-free cut. Pick any one vertex set of the cut and set the variables in those to true. Set the remaining variables to false. Now this assignment makes the expression NAE-satisfiable.
\\

Each triangle corresponding to a three variable clause has at least one true and one false vertex, thus making the clause NAE-Satisfiable. For any clause of the form $(x \vee \neg y)$, $x$ and $y$ are necessarily in the same set, thus making the clause NAE-Satisfiable.
\\

This proves the NP-Completeness of the triangle-free cut problem on $5$-colourable, bounded degree ($\leq 8$) graphs.
\end{proof}

\section{NP-Completeness of $k$-colourable Monotone-NAE-3SAT}

We prove the following:

\begin{thm}
$k$-colourable Monotone-NAE-3SAT is NP-Complete for $k \ge 5$. Moreover, it remains NP-Complete if each variable is restricted to appear
no more than $7$ times.
\end{thm}

\begin{proof}
It is easy to see that the above reduction to triangle-free cut carries over to the $5$-colourable Monotone NAE-3SAT in a natural manner (by making the vertices variables, and triangles the clauses) such that the graph has a triangle free cut if and only if the corresponding $5$-colourable Monotone NAE-3SAT expression is satisfiable. Moreover, each vertex of the graph appears in at most $7$ triangles ($2$ gadgets and $1$ clause triangle or $5$ triangles for a gadget vertex) and so each variable in the expression appears no more than $7$ times.
\\

Thus $5$-colourable Monotone NAE-3SAT is NP-Complete and remains so even under the restriction that each variable appears no more than 7 times.
\end{proof}

\section{Conclusion}
We have shown that the decision versions of $k$-colourable Monotone NAE-3SAT and Triangle-Free Cut problem for $k$-colourable graphs are NP-Complete for $k \ge 5$. The next steps would be to prove the NP-Completeness (or give polynomial time algorithms) of the corresponding search problems when $k = 3$ or $k = 4$.

\bibliographystyle{model1a-num-names}

\end{document}